\documentclass{llncs}

\usepackage{amsmath}\usepackage{amsfonts}
\usepackage{tikz}

\newcommand{\set}[1]{\left\{ #1 \right\}}
\newcommand{\setdef}[2]{\set{ #1 \, : \, #2}}
\newcommand{\struc}{\text{STRUC}}
\newcommand{\md}{\text{MOD}}
\newcommand{\tup}[1]{\langle #1 \rangle}

\newcommand{\cc}{\text{\sc 2CC}}

\newcommand{\sat}{\text{\sc Sat}}
\newcommand{\efsat}{\text{\sc QSat}_2}

\newcommand{\eunique}{\exists \exists!\text{\sc Sat}}

\newcommand{\eftsat}{\text{\sc Q3Sat}_2}

\newcommand{\eunsat}{\text{\sc QUnsat}_2}

\newcommand{\vcsat}{\text{\sc VCSat}}

\newcommand{\Sp}{\Sigma_2^{p}}

\newcommand{\sigst}{\mathcal{A}}
\newcommand{\sigfnc}{\sigma_{\text{cnf}}}
\newcommand{\sigfnd}{\sigma_{\text{dnf}}}

\newcommand{\fnc}{\text{\rm CNF}}

\newcommand{\tfnd}{3\text{DNF}}
\newcommand{\fnd}{\text{DNF}}
\newcommand{\bin}{\text{bin}}

\newcommand{\tr}{\text{\bf true}}

\newcommand{\suc}{\text{suc}}

\begin{document}

\title{A syntactic tool for proving hardness\\ in the Second Level of the \\ 
Polynomial-Time Hierarchy}
\titlerunning{Descriptive Complexity}  
%
\author{Edwin Pin\inst{1} \and Nerio Borges\inst{2}}
\authorrunning{Ivar Ekeland et al.} 
%
\tocauthor{Ivar Ekeland, Roger Temam, Jeffrey Dean, David Grove,
Craig Chambers, Kim B. Bruce, and Elisa Bertino}
\institute{Universidad Central de Venezuela,\\
Departamento de Matematica, \\
\email{edwin.pin@ciens.ucv.ve}\\ 
\and
Universidad Simon Bolivar,\\
Departamento de Matematicas Puras y Aplicadas,\\
\email{nborges@usb.ve}}

\maketitle              

\begin{abstract}
In \cite{medina}, Immerman and Medina initiated the search for syntactic tools to prove \textbf{NP}-completeness. In their work, amongst several results, they conjecture that the \textbf{NP}-completeness of a problem defined by the conjunction of a sentence in Existential Second Order Logic with a First Order sentence, necessarily imply the \textbf{NP}-completeness of the problem defined by the Existential Second Order sentence alone. This is interesting because if true it would justify the \textit{restriction} heuristic proposed in \cite{garey} which roughly says that in some cases one can prove \textbf{NP}-complete a problem $A$ by proving \textbf{NP}-complete a problem $B$ contained in $A$. 

Borges and Bonet \cite{borges,borges2,borges3} extend some results from Immerman and Me\-di\-na and they also prove for a host of complexity classes that the Immerman-Medina conjecture is true when the First Order sentence in the conjunction is universal \cite{borges3}. Our work extends that result to the Second Level of the Polynomial-Time Hierarchy.
\end{abstract}
\section{Introduction}
The concept of \textit{completeness} in a 
complexity class is one of the most relevants in Computational Complexity theory. 
The formulation of this concept, due to a result by Stephen Cook, 
led to the proposal of the quite famous open problem \textbf{P} versus \textbf{NP}. 
Cook proved that any \textbf{NP} problem can be efficiently reduced to the 
Boolean Satisfiability Problem commonly denoted as SAT \cite{cook}. 
The SAT problem is only the first of an extensive list of decision problems classified as 
\textbf{NP}-complete \cite{garey}. 
The major importance of this classification lies on the suspicion that 
the apparent contrast between \textbf{P} and \textbf{NP} 
is mostly due to the existence of \textbf{NP}-complete problems.

Another prominent result, proved by Ron Fagin \cite{fagin}, 
establishes that there is no need for a computational model 
(such as nondeterministic Turing machines) to define the complexity class \textbf{NP} 
but, instead, it can be defined using expressive resources provided by an appropriate language. 
Specifically, the set $\exists\textrm{SO}$ of existential-second-order sentences 
\textit{captures} \textbf{NP}, which means that every \textbf{NP} problem can be defined by a 
sentence in $\exists\textrm{SO}$ and every problem defined by a sentence in 
$\exists\textrm{SO}$ belongs to \textbf{NP} \cite{fagin}. 
Stockmeyer generalized this by defining new complexity classes whose union, 
known as the Polynomial-Time Hierarchy, is captured by second order logic 
\cite{stockmeyer}. Fagin's and Stockmeyer's results are fundamental to
Descriptive Complexity Theory. 
Immerman compiled in \cite{immerman} syntactic characterizations of most well-known 
complexity classes. 

In this line of research, Antonio Medina and Neil Immerman initiated the search for syntactic tools to prove \textbf{NP}-completeness \cite{medina}. In their work they conjecture that the \textbf{NP}-completeness of a problem defined by a sentence in the form $\Phi\land\varphi$ with $\Phi$ in $\exists\textrm{SO}$ and $\varphi$ in FO necessarily imply the \textbf{NP}-completeness of the problem defined by $\Phi$ alone. This is interesting because if true it would justify some instances of the \textit{restriction} heuristic proposed in \cite{garey} which roughly says that in some cases one can prove \textbf{NP}-complete a problem $A$ by proving \textbf{NP}-complete a problem $B$ contained in $A$. 

Borges and Bonet \cite{borges,borges2,borges3} extend some results from Immerman and Medina and they also prove for a host of complexity classes that the Immerman-Medina conjecture is true when $\varphi$ is a universal First Order sentence \cite{borges3}. Our work extends that result to the Second Level of the Polynomial-Time Hierarchy.

This paper is organized as follows: Section \ref{sec:prelim} overviews basic definitions, some of them from computational complexity theory but within the descriptive context. In Section \ref{sec:CompleteProblems} we will prove that most of the problems mentioned in the Appendix are complete under first-order projections, which is a necessary condition for the study of other concepts presented afterwards. Section \ref{sec:FurtherConcepts} introduces the fundamental ideas of this work, concepts and results that were fully developed in \cite{borges3}, superfluity included. In Section \ref{sec:main} we present the main results to show that superfluity is valid in the Second Level of the Polynomial-Time Hierarchy along with an application. We conclude the paper discussing how this investigation may continue in Section \ref{sec:Conclusions}.

\section{Preliminaries}\label{sec:prelim}
We shall first consider the descriptive approach to our objects of study. 
Although this is not the way it is commonly done, 
we are going to introduce complexity classes syntactically. 
Other concepts, as reducibility and completeness, 
will be understood in the context of Descriptive Complexity.

This section is included to keep this paper as self-contained as possible.
We follow closely the exposition in \cite{immerman2}
and most of its notation.

\subsection{Vocabularies and Languages}
Since our approach to computational complexity is descriptive,
we are going to see decision problems in terms of mathematical logic.


A \textit{relational vocabulary} 
is a tuple $\sigma = \langle R_1^{a_1},\ldots,R_r^{a_r}, c_1,\ldots, c_s \rangle$, 
where each $R_i$ is a \textit{relation symbol} 
with an associated positive integer $a_i$ called its \textit{arity}  
and each $c_j$ is a {\em constant symbol}. 
We do not consider function symbols but
this choice implies no loss of expressive power 
since functions can be defined as relations as well.
We also suppose that every vocabulary includes 
the so called
\textit{numeric} relation and constant symbols 
$\leq$, $\mbox{BIT}$, $\mbox{PLUS}$, $\mbox{TIMES}$, $\mbox{SUC}$, $0$, $1$, $\max$.

A $\sigma$-\textit{structure}, or simply a \textit{structure} 
if $\sigma$ is clear from context, 
is a tuple 
$\mathcal{A} = \langle |\mathcal{A}|,R_1^{\mathcal{A}},\ldots,R_r^{\mathcal{A}}, 
c_1^{\mathcal{A}},\ldots, c_s^{\mathcal{A}} \rangle$, 
where
\begin{itemize}
\item $|\mathcal{A}|$ is a nonempty set, called the \textit{universe} of $\mathcal{A}$,
\item $R_i^{\mathcal{A}}$ is an $a_i$-ary relation over $A$, that is, $R_i^{\mathcal{A}} \subseteq A^{a_i}$, and
\item $c_j^{\mathcal{A}}$ is an element of the universe.
\end{itemize}
A $\sigma$-structure provides interpretations for the symbols in $\sigma$.
The set of all finite $\sigma$-structures is denoted as $\textrm{STRUC}[\sigma]$.
When $\sigma$ contains only the relation symbol $R$ and no constant symbols
 we will denote $\textrm{STRUC}[\sigma]$ as
$\textrm{STRUC}[R]$. 

In this paper structures will represent instances of decision problems
hence we are going to adopt some conventions:
Every structure $\mathcal{A}$ has finite universe,
thus we say $\mathcal{A}$ is a \textit{finite} structure. 
The cardinality of $\mathcal{A}$ is denoted by $||\mathcal{A}||$. 
Since decision problems are closed under isomorphisms
we will assume that the universe of every structure $\mathcal A$ 
with cardinality $n>1$ is 
$|\mathcal A|=\{ 0,1,\ldots,n-1 \}$.
For short, we denote the set $\{ 0,1,\ldots,n-1 \}$ as $n$. 
The numeric relation and constant symbols are given their standard interpretations (see \cite{immerman2}).

We are going to consider formulas in first order logic and second order logic 
(for a detailed account, we refer the reader to \cite{immerman2}). 
A \textit{numeric} formula is a formula with only numeric symbols. 
A \textit{sentence} is a formula with no free variables. 
First order logic and second order logic are denoted by 
$\textrm{FO}$ and $\textrm{SO}$, respectively. 
Throughout this paper, we are going to consider restrictions of both logical languages.
For instance, the already mentioned $\exists\textrm{SO}$ which captures \textbf{NP}, 
or the set $\forall\textrm{FO}$ of universally-quantified-first-order sentences. 
We are going to refer to this kind of restrictions as logical languages also. 
When the vocabulary is worth to mention, we will write $\mathcal{L}[\sigma]$, 
where $\mathcal{L}$ is a logical language. 

When a $\sigma$-structure $\mathcal{A}$ satisfies a sentence $\varphi$ in $\mathcal{L}[\sigma]$,
we write $\mathcal{A} \models \varphi$. 
The set of all finite structures that satisfy $\varphi$ is denoted by 
$\textrm{MOD}[\varphi]$. 
A sentence $\varphi\in\mathcal{L}[\sigma]$ \textit{defines}
the decision problem $\Omega$ iff $\Omega=\textrm{MOD}[\varphi]$.
%
\subsection{Complexity Classes}
Let $\mathcal{L}$ be a logical language closed under disjunctions 
and closed under conjunctions with first-order formulas. 
The complexity class $\mathbf{C}$ captured by $\mathcal{L}$ 
is the set of all decision problems defined by 
sentences in $\mathcal{L}$ i.e. 
\begin{equation}
 \mathbf{C} := \{ \textrm{MOD}[\varphi] : \varphi \in \mathcal{L} \}.
 \label{eq:one}
\end{equation}

This notion of complexity classes follows from \cite{borges}, 
in which $\mathbf{C}$ is asked to be \textit{nice} 
\footnote{nice \cite{allender}, also known as 
\textit{syntactic} classes in \cite{papadimitriou},
are complexity classes having a ``universal'' complete language
i.e. a complete language
$\mathcal U\subset\set{0,1}^*$ such that $w\in U$ iff
it codifies a pair $(M,x)$ where $M$ is a TM
accepting input $x$ within the resource bounds defining
C.} 
closed under finite unions and also dependent 
on a family of proper complexity functions \cite{papadimitriou}. 
All the complexity classes mentioned throughout this paper satisfy these three conditions. 

Let $\textrm{SO}_k$ be the language of 
$\textrm{SO}$ sentences with at most $k$ alternations of quantifiers, 
starting with an existential one. So, for every natural number $k$,
$\textrm{SO}_k$ consists of all second order sentences with the form
\begin{equation}
\underbrace{\exists \vec R_1\forall\vec R_2\ldots \mathcal Q_k\vec R_k}
_{ k \textrm{ \scriptsize{quantifiers} }}\varphi
\end{equation}
where $\mathcal Q_k$ is existential if $k$ is odd and universal if $k$ is even,
each $\vec R_j$ is a tuple of relation variables, and $\varphi$
is a first order sentence.

We are going to pay special attention to the case $k=2$ i.e.
the language $\textrm{SO}_2$ of sentences with the form 
$\exists \vec R_1\forall\vec R_2\varphi$. 

The Polynomial-Time Hierarchy is defined by levels as follows: 
the level $0$ is \textbf{P} captured by SO-Horn \cite{immerman2} and is denoted by $\Sigma_0^p$. 
For $k \geq 1$, the kth level is
\begin{equation}
  \Sigma_k^p := \{ \textrm{MOD}[\varphi] : \varphi \in \textrm{SO}_k \}.
\end{equation}

The complementary class of $\Sigma_k^p$ is denoted by $\Pi_k^p$ i.e.
\begin{equation}
  \Pi_k^p := \{ \textrm{MOD}[\neg\varphi] : \varphi \in \textrm{SO}_k \}.
\end{equation}

As has been already stated, the first level agrees with the class \textbf{NP}. 
The Polynomial-Time Hierarchy is defined as the class
\begin{equation}
\mathbf{PH} := \bigcup_{k \geq 0} \Sigma_k^p=\bigcup_{k \geq 0} \Pi_k^p,
\end{equation}
which is captured by SO. 
 
To look at some examples of problems in $\mathbf{PH}$ see Appendix \ref{app:problems}.
Other problems in $\mathbf{PH}$ can be checked on \cite{papadimitriou}.

\subsection{Reductions and First-Order Queries}

We also need a precise syntactical notion 
for another important computational concept: reducibility. 
Let $A$ and $B$ be decision problems. 
Informally, $A$ is \textit{reducible to} $B$ if there is an easily
computable map $f$ from instances of $A$ to instances of $B$ such that
$x\in A\iff f(x)\in B$. Such a function together with a TM $M_B$ deciding $B$ 
yields an algorithm $M_A$ that decides $A$ and is not 
significantly harder than $M_B$,
thus we can conclude that $A$ is as hard to compute as $B$. 
Notice that the requirement of efficiency imposed over $f$ is quite important.

Formally,
let $\tau$ and $\sigma$ be two vocabularies with 
$\sigma = \langle R_1^{a_1},\ldots,R_r^{a_r}, c_1,\ldots, c_s \rangle$. 
Let $k$ be a positive integer and consider the tuple 
$I = \langle \varphi_0,\ldots,\varphi_r, \psi_1,\ldots, \psi_s \rangle$
of formulas in $\textrm{FO}[\tau]$
where $\varphi_0$ has arity $k$,
each $\varphi_j$ with $1\leq j\leq r$ has arity $ka_i$
and each $\psi_j$ with $1\leq j\leq s$ has arity $k$. 
$I$ defines a map 
\begin{equation}
\textrm{STRUC}[\tau] \longrightarrow \textrm{STRUC}[\sigma]
\end{equation}
that takes every $\tau$-structure $\mathcal{A}$, 
to a $\sigma$-structure $I(\mathcal{A})$ given by the tuple
\begin{equation}
\langle |I(\mathcal{A})|,R_1^{I(\mathcal{A})},\ldots,R_r^{I(\mathcal{A})},
 c_1^{I(\mathcal{A})},\ldots, c_s^{I(\mathcal{A})} \rangle,
\end{equation}
where
\begin{itemize}
\item $|I(\mathcal{A})|$ is the subset of $|\mathcal{A}|^k$
	defined by $\varphi_0(x_1,\ldots, x_k)$,
\begin{equation}
|I(\mathcal{A})| = \left\{ (b^1,\ldots,b^k)\in |\mathcal{A}|^k :
 \mathcal{A} \models \varphi_0(b^1,\ldots,b^k) \right\}.
\end{equation}
\item $R_i^{I(\mathcal{A})}$ is the subset of
 $|I(\mathcal{A})|^{a_i}$ defined by $\varphi_i$,
\begin{equation}
R_i^{I(\mathcal{A})} = \left\{ (\vec{b}_1, 
\ldots, \vec{b}_{a_i}) \in |I(\mathcal{A})|^{a_i} 
: \mathcal{A} \models \varphi_i(\vec{b}_1, \ldots, \vec{b}_{a_i}) \right\}.
\end{equation}
\item $c_j^{I(\mathcal{A})}$ 
is the only element of $|I(\mathcal{A})|$ satisfying $\psi_j$ i.e.
the only
$\vec{b} \in |I(\mathcal{A})|$ such that $\mathcal{A} \models \psi_j(\vec{b})$.
\end{itemize}

We call $I$ a $k$-\textit{ary first-order query} from $\textrm{STRUC}[\tau]$ 
to $\textrm{STRUC}[\sigma]$. 
Let's suppose that $A \subseteq \textrm{STRUC}[\tau]$ and 
$B \subseteq \textrm{STRUC}[\sigma]$. 
$I$ is a \textit{first-order reduction} from $A$ to $B$ if for every 
$\tau$-structure $\mathcal{A}$,
$$\mathcal{A} \in A \quad \iff \quad I(\mathcal{A}) \in B.$$

First-order reductions are quite interesting.
For a general treatment of their properties we refer the reader to \cite{allender}. 
A first-order query is called a \textit{first-order projection} 
(or \textit{fop}) if $\varphi_0$ is numeric and each 
$\varphi_i$ and $\psi_j$ is a first order formula in the form
\begin{equation}
\alpha_0 (\vec x) \vee (\alpha_1(\vec x) \wedge \lambda_1(\vec x)) \vee \cdots 
\vee (\alpha_{\ell}(\vec x) \wedge \lambda_{\ell}(\vec x)) 
\end{equation}
where
\begin{itemize} 
\item The $\alpha_k$'s are numeric and mutually exclusive
i.e. if $\mathcal A$ is a structure and $\vec u$ is a tuple of
elements from $|\mathcal A|$ with the appropriate length,
then $\mathcal A\models \alpha_j(\vec u)\Rightarrow
A\not\models \alpha_i(\vec u)$ for every $i\neq j$.
\item The $\lambda_k$'s are $\tau$-literals.
\end{itemize} 
Unless otherwise stated, our reductions are fops.  
We denote by $A \leq_{fop} B$ the fact that problem $A$ is reducible
to problem $B$. 
The binary relation $\leq_{fop}$ is transitive and reflexive 
thus it is a quasi order. It is not an order because
it is not antisymmetric.

Our idea of completeness in a complexity class depends on our notion of reduction.
A problem $B$ is \textit{hard} (via fops) in the complexity class $\mathbf{C}$ or $\mathbf{C}$-\textit{hard}
if $A \leq_{fop} B$ for every problem $A\in\mathbf{C}$.
We say that $B$ is $\mathbf{C}$-\textit{complete} (via fops) if $B\in\mathbf{C}$ and it is
$\mathbf{C}$-\textit{hard}.

There are other kinds of reductions e.g. polynomial time reductions
and log-space reductions and the corresponding completeness notions
in the different complexity classes. 
Clearly if we are working within a given complexity class,
the reductions allowed must not be more difficult than
the problems in the class. Thus when discussing completeness
in {\bf L} or {\bf NL}, for instance, we can not use polynomial time reductions.
Sometimes we will refer to hardness or completeness using other reductions than fops, but we are going to make it explicit.

A major reference in completeness via polynomial reductions in the second level of the Polynomial Hierarchy,
is the The compendium by Schaefer and Umans \cite{schaefer}, where several complete (via poly-reductions) problems
are listed. 
 
It is known that $\efsat$ is $\Sp$-complete 
via log-space reductions \cite{stockmeyer}. 
It is also known that SAT is complete via fops \cite{immerman2}, 
and an analogous construction can be considered to prove that $\efsat$ is $\Sp$-complete 
via fops. 
In \cite{marx,marx2}, it is proved that $\eunique$ and $2\textrm{CC}$ are $\Sp$-complete 
by reducing $\efsat$ to it.
 Those same reductions can be adapted to be fops. 
In \cite{berit} it is proved that $\vcsat$ and many other 
 value-and-cost problems are $\Sp$-complete for other reductions which are not fops. We are going to study this in detail in the next section.

\section{Some complete problems in $\Sp$}\label{sec:CompleteProblems}

This section is devoted to prove the following theorem. We will break down its demonstration into several propositions. 

\begin{theorem}\label{theorem-resumen}
The following problems are $\Sp$-complete:
\begin{enumerate}
	\item $\efsat$
    \item $\eunsat$
    \item $\eunique$
    \item $\cc$
\end{enumerate}
\end{theorem}

This four problems are known to be in $\Sp$. They are even known to be $\Sp$-complete for non-projective reductions. Thus it remains to show they are $\Sp$-hard via fops.

 Since $\leq_{fop}$ is a transitive relation we can prove that $B$ is hard in a complexity class {\bf C} taking a suitable {\bf C}-hard (or {\bf C}-complete) problem $A$ and reducing it to $B$. For this approach to work we need to prove a first problem $\Omega$ complete from scratch i.e. given a generic problem $\Pi$ in {\bf C} we have to prove that $\Pi\leq_{fop}\Omega$. In {\bf NP} this first complete problem for every reduction notion is usually $\sat$. In $\Sp$ it is $\efsat$. In \cite{stockmeyer} it is proved that $\efsat$ is $\Sp$-complete via log-space reductions. It can be shown that $\efsat$ is hard via fops with a proof similar to the one employed in \cite{immerman2} to show that $\sat$ is {\bf NP}-hard.

 \begin{proposition}\label{theorem-sat}
	$\efsat$ is $\Sp$-hard.
\end{proposition}

\begin{proof}
	Consider a vocabulary $\sigma$ and a problem $A \subseteq \struc[\sigma]\in \Sp$, so there is a sentence $\Phi$ in $\textrm{SO}_2[\sigma]$ such that $A = \md[\Phi]$. 
    We proceed to prove that $A \leq_{fop} \efsat$.
    Since $A\in\Sp$ we can assume $\Phi$ has the form:
	\begin{equation}\label{eq:efsat}
	\Phi \equiv \exists S_1^{a_1} \cdots S_g^{a_g} \forall T_1^{b_1} \cdots T_h^{b_h} \exists x_1\cdots x_c \varphi,
	\end{equation}
	with
	$$\varphi(x_1,\ldots,x_c) \equiv \bigvee_{i=1}^r D_i(x_1,\ldots,x_c),$$
	where each $D_i$ is an implicant $(L_1 \wedge \ldots \wedge L_{\ell_i})$. Each $L$ is a literal. If $L$ is an existentially (universally) quantified relation $S_1,\ldots,S_g$ ($T_1,\ldots,T_h$) or its negation we say it is an \textit{existential} (\textit{universal}) literal. Notice we require $\varphi(x_1,\ldots,x_c)$ to be in $\fnd$. We can also assume that each $D_i$ has at most one literal from $\sigma$. In the following, notation $L(\vec{x})$ means that literal $L$ is evaluated in the $c$-tuple $\vec{x}$ as $\varphi$ indicates, not that $L$ is a $c$-ary relation. 
	
	Suppose $\mathcal{A}$ is an instance of $A$, we must map it to a boolean formula $\rho(\mathcal{A})$ satisfying
	\begin{equation}\label{eq:sii-efso}
	\mathcal{A} \in A \quad \iff \quad \rho(\mathcal{A}) \in \efsat.
	\end{equation}
	We will use the sentence $\exists x_1\cdots x_c \varphi$ to do that. We describe $\rho(\mathcal{A})$ with the vocabulary $\sigfnd=\tup{E^1,Q^2,M^2}$. The relation symbol $E$ is intended to identify existential variables, $Q(x,y)$ (resp. $M(x,y)$) means that variable $y$ occurs positively (negatively) in implicant $x$.
	
	We identify the universe of $\rho(\mathcal{A})$ with a subset of $|\mathcal{A}|^k$ where $k = \log(m) + c$ and $m = \max\{g+h,r\}$. Each element  $\vec{x} \in |\rho(\mathcal{A})|$ will be regarded as the concatenation of two tuples $\vec{x}_1$ and $\vec{x}_2$ of lengths $|\vec{x}_1|=\log(m)$ and $|\vec{x}_2|=c$. Some elements of $|\rho(\mathcal{A})|$ represent subformulas (implicants) and atomic formulas in $\Phi$ as follows: 
	
	\begin{itemize}
		\item implicant $D_i(\vec{x}_2)$ is interpreted as the tuple $\vec{x}_1 \vec{x}_2$ where $\vec{x}_1$ is the binary codification of index $i$;
        \item literal $S_i(\vec{x}_2)$ is interpreted as the Boolean variable $\vec{x}_1 \vec{x}_2$ where $\vec{x}_1$ is the binary codification of index $i$ (which is a number between 1 and $g$);
		\item literal $T_j(\vec{x}_2)$ is interpreted as the Boolean variable $\vec{x}_1 \vec{x}_2$ where $\vec{x}_1$ is the binary codification of the number $g+j$ to represent index $j$.
	\end{itemize}
    
    Notice that many elements in the universe of $\rho(\mathcal{A})$ might refer to an implicant and a Boolean variable simultaneously, but this is not a problem at all, because interpretations of symbols $E$, $Q$ and $M$ will be quite clear in context.
    
    A structure $\mathcal{A}$ is a positive instance of problem $A$ iff for some interpretation of the literals $S$ there is an index $i = 1,\ldots, r$ and a $c$-tuple $\vec{x}_2$ such that $\mathcal{A} \models D_i(\vec{x}_2)$, no matter how the universal literals $T$ are interpreted. The implicants of $\rho(\mathcal{A})$ are determined by those $D_i(\vec{x}_2)$ such that its satisfiability depends necessarily on the existential and universal literals. Based on these ideas the relations $E^{\rho(\mathcal{A})}$, $Q^{\rho(\mathcal{A})}$ and $M^{\rho(\mathcal{A})}$ will be constructed.
	
	Set $E^{\rho(\mathcal{A})}$ is easily described by a numeric first-order formula. The tuple $\vec{x}_1\vec{x}_2$ is in $E$ (i.e. it represents an existential Boolean variable) iff $\vec{x}_1$ is the binary codification of some number $i \leq g$. In short, $E$ is defined by the formula
    \begin{equation}
    \varphi_E(\vec{x}_1\vec{x}_2) \equiv (\exists i) (i \leq g \wedge \vec{x}_1 = \bin (i)),
    \end{equation}
	where $\bin(i)$ denotes the binary codification of $i$ in $\log(m)$-bits and $\vec{x}_1 = \bin (i)$ is the bit-equality. 
    
    The set $Q^{\rho(\mathcal{A})}$ requires some further analysis. Suppose the implicant $D_i(\vec{x}_2)$ contains the sub-formula
	$\big(\alpha \wedge R(\vec{x}_2) \wedge L(\vec{x}_2)\big)$, where $\alpha$ is the conjunction of every numeric literal appearing in $D_i$, $R(\vec{x}_2)$ is the only $\sigma$-literal in $D_i$ and $L$ is any positive existential or universal literal in $D_i$. If $\mathcal{A} \not\models \alpha \wedge R(\vec{x}_2)$, there is no need to refer to $L(\vec{x}_2)$ because even if $\mathcal{A}$ satisfies it,
	$\mathcal{A} \not\models D_i(\vec{x}_2)$, so the implicants where this happens will be discarded. Now, the pair $(\vec{x}_1\vec{x}_2,\vec{y}_1\vec{y}_2)$ is in $Q$ iff $\vec{x}_1\vec{x}_2$ codifies an implicant $D_i(\vec{x}_2)$ such that its positive literals $L(\vec{x}_2)$ might be relevant for its satisfiability. In short, some part of $Q$ is determined by the disjunction of all the formulas
    \begin{equation}\label{eq:sat-red}
    (\vec{x}_1 = \bin(i)) \wedge (\vec{y}_1 = \bin([\ell])) \wedge (\vec{x}_2 = \vec{y}_2) \wedge \alpha \wedge R(\vec{y}_2),
    \end{equation}
	where the value $\ell$ is the index corresponding to literal $L$ as a relational variable and $[\ell]$ is $\ell$ if $L$ is existential or it is $g+\ell$ if $L$ is universal. The other part of $Q$ is determined by all the implicants in $\Phi$ that don't contain $\sigma$ relations i.e. formulas quite similar to (\ref{eq:sat-red}) but without the atom $R(\vec{y}_2)$. Denote by $\varphi_Q$ the disjunction of every formula described for $Q$.
	
	Similarly, we can construct a formula $\varphi_M$ to describe relation $M^{\rho(\mathcal{A})}$, except that in this case $L$ represents a negative existential or universal literal in a certain implicant.
	
	Notice that every formula mentioned so far is numeric or projective. Furthermore, reduction $\rho = \lambda_{\vec{x},\vec{y}}\langle \tr , \varphi_E,\varphi_Q,\varphi_M \rangle$ was constructed to satisfy condition (\ref{eq:sii-efso}).    \hfill$\square$
\end{proof}

As relation $\leq_{fop}$ is transitive, to evaluate the $\Sp$-completeness of a problem $B$ it is enough to prove that $\efsat \leq_{fop} B$.

\begin{example}\label{example-eunsat-complete}
	Known properties of Boolean formulas allow us to construct a natural reduction from $\efsat$ to $\eunsat$. Let $\sigst$ be a $\sigfnd$-structure and let $\rho(\sigst)$ be a $\sigfnc$-structure defined as follows:
	\begin{itemize}
		\item $|\rho(\mathcal{A})|=|\mathcal{A}|$;
		\item $E^{\rho(\mathcal{A})}=E^{\mathcal{A}}$, described by projective formula $\varphi_E(x) \equiv E(x)$;
		\item $P^{\rho\mathcal{(A)}}=M^{\mathcal{A}}$, described by projective formula $\varphi_P(x,y) \equiv M(x,y)$;
		\item $N^{\rho\mathcal{(A)}}=Q^{\mathcal{A}}$, described by projective formula $\varphi_N(x,y) \equiv Q(x,y)$.
	\end{itemize}
	
    Notice that the last two items means that $\rho(\sigst)$ is the negation of $\sigst$, written by De Morgan' law as a CNF Boolean formula, which is a new structure obtained from $\sigst$ through projective formulas. Now, $\rho=\lambda_{xy}\langle \tr,\varphi_E,\varphi_P,\varphi_N \rangle$ is a projection from $\efsat$ to $\eunsat$. \hfill$\square$
\end{example}

In the following examples we show that reductions propose by Daniel Marx in \cite{marx,marx2} are in fact projections. 

\begin{example}\label{example-unique-complete}
	A reduction from $\eunsat$ to $\eunique$. The property that allow us to prove the latter problem is $\Sp$-hard is the following: the Boolean formula  $\phi(y_1,\ldots,y_m)$ is unsatisfiable iff
	\begin{equation}\label{eq:eunique2}
	(z \vee \phi(y_1,\ldots,y_m)) \wedge (\neg z \vee y_1) \wedge \cdots \wedge (\neg z \vee y_m)
	\end{equation}
	has only one truth valid assignment (specifically, the one that assigns the value $\tr$ to every variable $y_i$ and to the new variable $z$). Notice that if $\phi(y_1,\ldots,y_m)$ is a $\fnc$ formula then we can assume (\ref{eq:eunique2}) is also a $\fnc$ formula, because the new variable $z$ can be distributed in every clause of $\phi$.
	
	Let $\mathcal{A}$ be a $\sigfnc$-structure as an instance of $\eunsat$. We need to construct another $\sigfnc$-estructura $\rho(\mathcal{A})$, such that $$\mathcal{A} \in \eunsat \iff \rho(\mathcal{A}) \in \eunique.$$ If $\mathcal{A}$ is the codification of a Boolean formula $\phi$, then $\rho(\mathcal{A})$ will be the codification of (\ref{eq:eunique2}). Suppose the universe of $\mathcal{A}$ is $n$ and define
	\begin{itemize}
		\item $|\rho(\mathcal{A})|=\{(i,j) \in n^2 : i=0 \vee i=1\}=2n$;
		\item the pairs $(0,y)$ are the interpretations of the variables $y$ of $\mathcal{A}$;
		\item the pair $(1,0)$ is the interpretation of the new variable $z$;
		\item the first clauses of $\rho(\mathcal{A})$ are the same of $\mathcal{A}$ except that in each one of them the variable $z$ is included;
        \item every variable $y$ defines a new clause on $\rho(\mathcal{A})$: $(\neg z \vee y)$ is a clause of $\rho(\mathcal{A})$ iff $y$ is an universal variable of $\mathcal{A}$, otherwise, the tautology $(\neg y \vee y)$ is the corresponding clause.		
	\end{itemize}
	
	Considering all these conditions we described explicitly the structure $\rho(\mathcal{A})$. The numeric formula $\psi_0(x,y) \equiv (x=0 \vee x=1)$ described the universe.
	
	The set $E^{\rho(\mathcal{A})}$ is described by the formula
	\begin{equation*}
	\psi_E(x,y) \equiv \big[ x=1 \wedge y \neq 0 \big] \vee \big[ x=0 \wedge E(y) \big].
	\end{equation*}
	
	The set $P^{\rho(\mathcal{A})}$ is described by the formula
	\begin{align*}
	\psi_P(x,y,z,w) \equiv &\big[ x=0 \wedge z=1 \wedge w=0 \big] \vee
	\big[ x=0 \wedge z=0 \wedge  P(y,w)  \big] \vee \\
	&\big[ x=1 \wedge z=0 \wedge y=w \wedge \neg E(y) \big] \vee
	\big[ x=1 \wedge z=1 \wedge y=w \wedge E(y) \big].
	\end{align*}
	
	The set $N^{\rho(\mathcal{A})}$ is described by the formula
	\begin{align*}
	\psi_N(x,y,z,w) \equiv &\big[ x=0 \wedge z=0 \wedge  N(y,w)  \big] \vee
	\big[ x=1 \wedge z=1 \wedge w=0 \wedge \neg E(y) \big] \vee \\
	&\big[ x=1 \wedge z=1 \wedge y=w \wedge E(y) \big].
	\end{align*}
	
	The last implicant in both $\psi_P$ and $\psi_N$ is an auxiliary condition, that allows to see every element of $|\rho(\sigst)|$ as the index of some clause. By construction, $\rho=\lambda_{xyzw}\langle \psi_0,\psi_E,\psi_P,\psi_N \rangle$ is a projection from $\eunsat$ to $\eunique$. \hfill$\square$
\end{example}

	

\begin{example}\label{example-2cc-complete}
	A reduction from $\efsat$ to $\cc$.\footnote{Actually, Marx reductions are defined from $\eftsat = \efsat \cap \tfnd$, where $\tfnd$ is the set of Boolean formulas in $\fnd$ with no more than three literals per implicant. This problem is also $\Sp$-complete \cite{stockmeyer}. The proves of $\cc$ and $\eunique$ completeness don't depend on the number of literals in every implicant, that's why we decided to work with $\efsat$.} Let $\sigst=\langle n,E,Q,M \rangle$ be an instance of $\efsat$ that codifies a Boolean formula $\phi$. Consider the graph $\mathcal{G}_{\phi}$ with the following characteristics:
	\begin{itemize}
		\item $\mathcal{G}_{\phi}$ has $6n$ nodes;
		\item For each Boolean variable $x_i$ of $\phi$ there are four kinds of labels in $\mathcal{G}_{\phi}$:\\ $x_i$, $\bar{x_i}$, $x_i'$ y $\bar{x_i}'$;
		\item For each implicant $p_i$ of $\phi$ there are two kinds of labels in $\mathcal{G}_{\phi}$: $p_i$ y $p_i'$.
	\end{itemize}
	This is all concerning the universe of $\mathcal{G}_{\phi}$. Regarding the edges we established the following conditions:
	\begin{itemize}
		\item Nodes $x$ and $\bar{x}$ are adjacent to $x'$ y $\bar{x}'$ respectively;
		\item If $x$ is an existential variable of $\phi$, then $x'$ is adjacent to $\bar{x}'$;
		\item Nodes $p_i$ y $p_i'$ are adjacent, just like $p_i'$ and $p_{i+1}$. In other words, the sequence $p_1-p_1'-p_2-p_2'-\cdots-p_n-p_n'$ is a path in $\mathcal{G}_{\phi}$;
		\item If $x$ isn't an existential variable of $\phi$, then the nodes $x'$ and $\bar{x}'$ are adjacent to $p_n'$;
		\item The set of nodes $x$ and $\bar{x}$ induces the \textit{greatest} graph not containing the edges $\{ x,\bar{x} \}$ for every node $x$;
		\item Node $x$ is adjacent to $p$, if $x$ appears in implicant $p$ of $\phi$;
		\item Node $\bar{x}$ is adjacent to $p$, if $\neg x$ appears in implicant $p$ of $\phi$;
		\item If none of the literals $x$ or $\neg x$ appears in implicant $p$, then nodes $x$ and $\vec{x}$ are adjacent to $p$.
	\end{itemize}
	
	In the image it is shown in detail the graph $\mathcal{G}_{\phi}$ for a particular formula $\phi$. The dashed line surrounding variables $x$ and $\bar{x}$ represents the graph induced by these variables, as established in the fifth item from the last list. 
	
	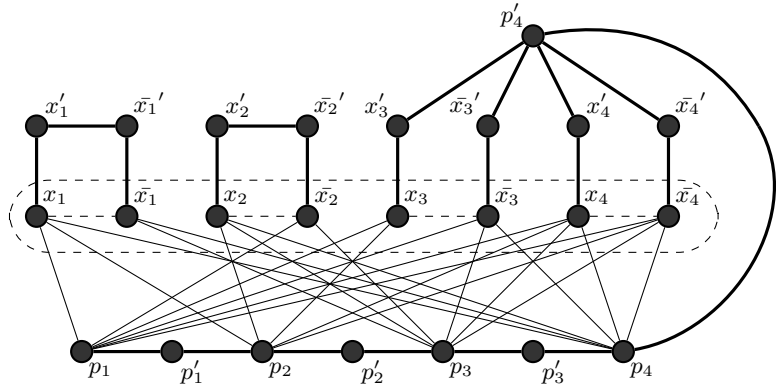
\begin{figure}[h]\label{fig:marx-graph}
    \begin{center}
		\tikzstyle{every node}=[circle, draw, fill=black!80,inner sep=0pt, minimum width=8pt]
		\begin{tikzpicture}[thick,scale=0.6]
		\node (x1) at (0,0) [label = above right : \small$x_1$] {};
		\node (nx1) at (2,0) [label = above right : \small$\bar{x_1}$] {};
		\node (x12) at (0,2) [label = above right : \small$x_1'$] {};
		\node (nx12) at (2,2) [label = above right : \small$\bar{x_1}'$] {};
		\node (x2) at (4,0) [label = above right : \small$x_2$] {};
		\node (nx2) at (6,0) [label = above right : \small$\bar{x_2}$] {};
		\node (x22) at (4,2) [label = above right : \small$x_2'$] {};
		\node (nx22) at (6,2) [label = above right : \small$\bar{x_2}'$] {};
		\node (x3) at (8,0) [label = above right : \small$x_3$] {};
		\node (nx3) at (10,0) [label = above right : \small$\bar{x_3}$] {};
		\node (x32) at (8,2) [label = above left : \small$x_3'$] {};
		\node (nx32) at (10,2) [label = above left : \small$\bar{x_3}'$] {};
		\node (x4) at (12,0) [label = above right : \small$x_4$] {};
		\node (nx4) at (14,0) [label = above right : \small$\bar{x_4}$] {};
		\node (x42) at (12,2) [label = above right : \small$x_4'$] {};
		\node (nx42) at (14,2) [label = above right : \small$\bar{x_4}'$] {};
		\node (p1) at (1,-3) [label = below right : \small$p_1$] {};
		\node (p12) at (3,-3) [label = below right : \small$p_1'$] {};
		\node (p2) at (5,-3) [label = below right : \small$p_2$] {};
		\node (p22) at (7,-3) [label = below right : \small$p_2'$] {};
		\node (p3) at (9,-3) [label = below right : \small$p_3$] {};
		\node (p32) at (11,-3) [label = below right : \small$p_3'$] {};
		\node (p4) at (13,-3) [label = below right : \small$p_4$] {};
		\node (p42) at (11,4) [label = above left : \small$p_4'$] {};
		\draw[very thick] (x1) -- (x12) -- (nx12) -- (nx1);
		\draw[very thick] (x2) -- (x22) -- (nx22) -- (nx2);
		\draw[very thick] (x3) -- (x32);
		\draw[very thick] (nx32) -- (nx3);
		\draw[very thick] (x4) -- (x42);
		\draw[very thick] (nx42) -- (nx4);
		\draw[very thick] (p42) -- (x32);
		\draw[very thick] (nx32) -- (p42);
		\draw[very thick] (p42) -- (x42);
		\draw[very thick] (nx42) -- (p42);
		\draw[very thick] (p1) -- (p12) -- (p2) -- (p22) -- (p3) -- (p32) -- (p4);
		\draw[dashed,rounded corners=15pt,thin] (-0.6,-0.8) rectangle (15.1,0.8);
		\draw[dashed,thin] (x1) -- (nx1);
		\draw[dashed,thin] (x2) -- (nx2);
		\draw[dashed,thin] (x3) -- (nx3);
		\draw[dashed,thin] (x4) -- (nx4);
		\draw[very thick] (p4) [out=10, in=-60] to (16,2);
		\draw[very thick] (p42) [out=10, in=120] to (16,2);
		\draw[thin] (x1) -- (p1);
		\draw[thin] (nx2) -- (p1);
		\draw[thin] (x3) -- (p1);
		\draw[thin] (nx3) -- (p1);
		\draw[thin] (x4) -- (p1);
		\draw[thin] (nx4) -- (p1);
		\draw[thin] (x1) -- (p2);
		\draw[thin] (x2) -- (p2);
		\draw[thin] (x3) -- (p2);
		\draw[thin] (x4) -- (p2);
		\draw[thin] (nx4) -- (p2);
		\draw[thin] (nx1) -- (p3);
		\draw[thin] (x2) -- (p3);
		\draw[thin] (nx2) -- (p3);
		\draw[thin] (nx3) -- (p3);
		\draw[thin] (x4) -- (p3);
		\draw[thin] (nx4) -- (p3);
		\draw[thin] (x1) -- (p4); 
		\draw[thin] (nx1) -- (p4);
		\draw[thin] (x2) -- (p4);
		\draw[thin] (nx3) -- (p4);
		\draw[thin] (x4) -- (p4);  
		\draw[thin] (nx4) -- (p4); 
		\end{tikzpicture}
		\caption{Graph for the Boolean formula $\phi_1 \equiv (x_1 \wedge \neg x_2) \vee (x_1 \wedge x_2 \wedge x_3 ) \vee (\neg x_1 \wedge \neg x_3) \vee (x_2 \wedge \neg x_3)$, with $x_1$ y $x_2$ as the existential variables.}
    \end{center}
	\end{figure}
	
	The function $\rho:\struc[\sigfnd'] \rightarrow \struc[G^2]$  given by $\rho(\sigst)=\mathcal{G}_{\phi}$ satisfies
	$$\sigst \in \efsat \quad \iff \quad \mathcal{G}_{\phi} \in \cc$$
    where $\phi$ is the Boolean formula codified by $\sigst$.
    The details of these fact can be seen in \cite{marx2}, we will only show why $\rho$ is a projection. The arity of $\rho$ is fixed as $k=4$. The universe of $\mathcal{G}_{\phi}$ is
	$$|\mathcal{G}_{\phi}|=\big\{ (i,j,k,x)\in n^{4} : ijk=\bin(m) \text{ para algun } m=1,\ldots,6 \big\}.$$
	
	In this example $\bin(m)$ is the three-bit-binary representation of $m$ and the expression $ijk=\bin(m)$ is the bit equality. The numeric first order formula describing the universe is
	$$\varphi_0(i,j,k,x) \equiv  (ijk=\bin(1))  \vee \cdots \vee (ijk=\bin(6)) .$$
	
	With this formula we are stating that $|\mathcal{G}_{\phi}|$ has exactly $6n$ elements. For each $x\in n$, the tuples $(0,0,1,x)$, $(0,1,0,x)$, $(0,1,1,x)$ and $(1,0,0,x)$ represent the nodes $x$, $x'$, $\bar{x}'$ y $\bar{x}$ respectively, while for each variable $p\in n$ (now as an index for implicants), $(1,0,1,p)$ and $(1,1,0,p)$ represent the nodes $p$ and $p'$ respectively. The formula $\varphi_G$ describing the set of edges will be the disjunction of all the following first-order formulas. 
	
	Nodes $x$ and $\bar{x}$ are adjacent to $x'$ and $\bar{x}'$ respectively:
	$$\big(x=p \wedge i_1j_1k_1=001 \wedge i_2j_2k_2=010\big) \vee \big(x=p \wedge i_1j_1k_1=011 \wedge i_2j_2k_2=100\big).$$
	
	If $x$ is an existential variable, then the node $x'$ is adjacent to $\bar{x}'$:
	$$\big(E(x) \wedge x=p \wedge i_1j_1k_1=010 \wedge i_2j_2k_2=011\big).$$
	
	The sequence $p_1-p_1'-p_2-p_2'-\cdots-p_n-p_n'$ is a path in $\mathcal{G}_{\phi}$:
	$$\big(x=p \wedge i_1j_1k_1=101 \wedge i_2j_2k_2=110\big) \vee \big(\suc(x,p) \wedge i_1j_1k_1=110 \wedge i_2j_2k_2=101\big).$$
	
	If $x$ isn't an existential variable, then nodes $x'$ y $\bar{x}'$ are adjacent to $p_n'$:
	\begin{align*}
	\big(\neg E(x) \wedge p=\max &\wedge i_1j_1k_1=010 \wedge i_2j_2k_2=110\big) \vee \\ &\big(\neg E(x) \wedge p=\max \wedge i_1j_1k_1=011 \wedge i_2j_2k_2=110\big).
	\end{align*}
	
	The set of nodes $x$ and $\bar{x}$ induces the greatest graph not containing the edges $\{ x,\bar{x} \}$:
	\begin{align*}
	\big( x\neq p \wedge& i_1j_1k_1=i_2j_2k_2=001 \big) \vee \big( x\neq p \wedge i_1j_1k_1=i_2j_2k_2=100 \big) \vee \\
	&\qquad\qquad \big( x\neq p \wedge i_1j_1k_1=001 \wedge i_2j_2k_2=100 \big).
	\end{align*}
	
	The last three edge conditions get compiled by the following formula:
	$$\big(i_1j_1k_1=001 \wedge i_2j_2k_2=101 \wedge \neg M(x,p)\big) \vee 
	\big(i_1j_1k_1=100 \wedge i_2j_2k_2=101 \wedge \neg Q(x,p)\big).$$
	
	Notice that all these formulas are projective and the numeric parts are mutually exclusive. Finally, the interpretation $\rho=\lambda_{ijkx}\langle \varphi_0,\varphi_G \rangle$ is a projection from $\efsat$ to $\cc$.   \hfill$\square$
\end{example}

\section{Further Concepts}\label{sec:FurtherConcepts}
Suppose $\Psi$ is a sentence in $\exists\textrm{SO}$.
According to \cite{medina} a first-order sentence $\varphi$ is \textit{superfluous} 
if there exists a fop $\rho$ from SAT to $\textrm{MOD}[\Psi \wedge \varphi]$ 
such that $\rho(\mathcal{A}) \models \varphi$
for every structure $\mathcal{A}$ representing a CNF Boolean formula.
An immediate consequence of the superfluity of $\varphi$ is
that the \textbf{NP}-completeness of $\textrm{MOD}[\Psi\land\varphi]$ 
implies that $\textrm{MOD}[\Psi]$ is \textbf{NP}-complete as well.

\begin{proposition}\cite{medina}\label{prop:medina}
If the conjunction $\Psi \land \varphi$ defines an \textbf{NP}-complete
with $\Psi$ a sentence in $\exists\textrm{SO}$
and $\varphi$ a superfluous sentence in $\textrm{FO}$
then $\textrm{MOD}[\Psi]$ is \textbf{NP}-complete
\end{proposition}

Classes of structures definable in first-order logic 
are strictly contained in {\bf L}, thus the expressive power
of first-order logic is strictly less than the expressive
power of existential second-order.
It is reasonable then to conjecture 
that the hypothesis of $\varphi$ being superfluous in
Proposition \ref{prop:medina} is not necessary:
\begin{conjecture}\cite{medina}\label{con:medina}
If the conjunction $\Psi \land \varphi$ defines an \textbf{NP}-complete
with $\Psi$ a sentence in $\exists\textrm{SO}$
and $\varphi$ a sentence in $\textrm{FO}$
then $\textrm{MOD}[\Psi]$ is \textbf{NP}-complete. 
\end{conjecture}

There is a partial answer to Conjecture \ref{con:medina} in \cite{borges3},
as a consequence of a stronger result (see Theorem \ref{theorem:borges}).
\begin{proposition}\cite{borges3}\label{prop:borges}
If the conjunction $\Psi \land \varphi$ defines an \textbf{NP}-complete
with $\Psi$ a sentence in $\exists\textrm{SO}$
and $\varphi$ a sentence in $\forall\textrm{FO}$
then $\textrm{MOD}[\Psi]$ is \textbf{NP}-complete
\end{proposition}
 
The following definitions and results are necessary to prove 
Theorem \ref{theorem:borges} and its corollary
Proposition \ref{prop:borges}. 

\subsection{Superfluity, Consistency and Universality}

Let $\sigma$ and $\tau$ be two vocabularies, $\mathcal{L}$ a logic, $\mathbf{C}$ the complexity class captured by $\mathcal{L}$, 
$\mathcal{L}'$ a fragment of $\mathcal{L}$ and $\rho:\textrm{STRUC}[\sigma] \rightarrow \textrm{STRUC}[\tau]$ a fop.
\begin{enumerate}
\item A sentence $\varphi \in \mathcal{L}'[\tau]$ is \textit{superfluous with respect to} $\rho$ if 
$\rho(\mathcal{A}) \models \varphi$ for every $\mathcal{A}  \in \textrm{STRUC}[\sigma]$.
\item $\varphi \in \mathcal{L}'$ is \textit{superfluous with respect to} $\mathcal{L}$ 
if for every sentence $\Psi \in \mathcal{L}$, the $\mathbf{C}$-completeness of $\textrm{MOD}[\Psi \wedge \varphi]$ 
implies the $\mathbf{C}$-completeness of $\textrm{MOD}[\Psi]$.
\item $\mathcal{L}'$ is \textit{superfluous with respect to} $\mathcal{L}$ (or $\mathbf{C}$) 
if every sentence $\varphi \in \mathcal{L}'$ is superfluous with respect to $\mathcal{L}$.
\end{enumerate}

Medina's conjecture can be paraphrased as 
FO is superfluous with respect to \textbf{NP}. 
To established the results of $\forall\textrm{FO}$ as a superfluous logic we need to 
introduce the notion of consistency of formulas and universality of problems. 

\begin{definition}
Let $\varphi(\vec{x})$ be a formula in $\textrm{FO}[\sigma]$, 
$n$ be a natural number, and $\vec{u} \in n^k$, where $k$ is the length of the first-order-variable tuple $\vec{x}$. 
We say that $\langle \varphi(\vec{x}),\vec{u} \rangle$ is $n$-\textit{consistent} 
if there is a $\sigma$-structure $\mathcal{A}$ such that $||\mathcal{A}|| = n$ and 
$\mathcal{A} \models \varphi(\vec{u})$. If $S \subseteq \textrm{STRUC}[\sigma]$, 
we say that $\langle \varphi(\vec{x}),\vec{u} \rangle$ is $n$-\textit{consistent in} $S$
if there is a structure $\mathcal{A} \in S$ such that $||\mathcal{A}|| = n$ and $\mathcal{A} \models \varphi(\vec{u})$. 
If there is no risk of confusion, we abbreviate by just saying that $\varphi(\vec{u})$ is $n$-consistent (in $S$).
\end{definition}

\begin{definition}\label{def:universal}
Let's suppose now that $\sigma = \langle R_1^{a_1},\ldots,R_r^{a_r}, c_1,\ldots, c_s \rangle$ and $S$ the same as before. Let $n$ and $t$ be two natural numbers. 

\begin{enumerate}

\item $S$ is $(n,0)$-\textit{universal} 
  if for every $m \geq n$ and every sequence $b_1,\ldots,b_s \in m$ 
  there is a structure $\mathcal{A} \in S$ with $\|\mathcal{A}\|=m$ and such that 
  $\mathcal{A} \models (c_1=b_1) \wedge \cdots \wedge (c_s=b_s)$.
\item $S$ is $(n,t)$-\textit{universal} if for every $m \geq n$, every sequence of $\sigma$-literals $L_1,\ldots,L_t$ (that is, 
  $L_j(\vec{x})$ is equal to $R_{i_j}(\vec{x})$ or $\neg R_{i_j}(\vec{x})$), 
  every sequence of tuples $\vec{u}_1,\ldots,\vec{u}_t$ with $\vec{u}_j \in m^{a_{i_j}}$ 
  and every sequence $b_1,\ldots,b_s \in m$, the $m$-con\-sis\-tency of
\begin{equation}
\varphi(\vec{u}_1,\ldots,\vec{u}_t,b_1,\ldots,b_s) \equiv \bigwedge L_j(\vec{u}_j) \wedge \bigwedge (c_k=b_k) \label{eq-universal}
\end{equation}
		implies its $m$-consistency in $S$ (that is, if there are models of (\ref{eq-universal}) of cardinality $m$, at least one belongs to $S$).
\end{enumerate}
\end{definition}
Universal problems are originally introduced in \cite{borges3} where they were called \textit{uniform}. 

The following properties are direct consequences of the later definition. 

\begin{lemma} \label{lemma:universal}
If $S \subseteq \textrm{STRUC}[\sigma]$ is $(n,k)$-universal, then it is also $(n,k-1)$-universal and $(n+1,k)$-universal.
If $S$ is $(n,k)$-universal and $S \subseteq T$, then $T$ is $(n,k)$-universal. 
\end{lemma}

In \cite{borges3} it is proved the universality of many well-known problems. 
In the next section we will prove that 2CC and its complement are also universal problems.

\begin{definition}\label{def:family}
A family $\mathcal{F}$ of problems over a vocabulary $\sigma$ is \textit{complete and universal} for a complexity class $\mathbf{C}$ if 
\begin{enumerate}
\item every problem in $\mathcal{F}$ is $\mathbf{C}$-complete;
\item There is a sequence $\{ n_k \}_{k\geq 0}$ and a natural number $m$ such that for every $k\geq m$ there is a $(n_k,k)$-universal 
  problem $S_{n_k}$ in $\mathcal{F}$ that contains all the $\sigma$-structures $\mathcal{A}$ with $||\mathcal{A}|| < n_k$.
\end{enumerate}
\end{definition}

\begin{theorem}\label{theorem:borges}\cite{borges3}
Let $\mathbf{C}$ be a complexity class captured by $\mathcal{L}$ with $\mathrm{FO} \subseteq \mathcal{L}$. 
If $\mathbf{C}$ contains a complete and universal family $\mathcal{F}$, 
then $\forall\mathrm{FO}$ is superfluous with respect to $\mathbf{C}$.
\end{theorem}

In \cite{borges3} it is also proved that universality and completeness are not intrinsically related concepts, 
due to the fact that there are {\bf NP}-complete problems that are not $(n,k)$-universal for sufficiently large $k$.

\section{Superfluity in the Second Level of PH}\label{sec:main}
\subsection{Superfluity in $\Sigma_2^p$}

We want to prove that $\forall\mathrm{FO}$ is superfluous with respect to $\Sigma_2^p$ applying 
\textbf{Theorem \ref{theorem:borges}}. Thus we have to prove that $\Sigma_2^p$ 
contains a complete and universal family $\mathcal{F}$. 
We will generate that family from 2CC.

\begin{definition}
If $S$ is a problem over a vocabulary $\sigma$, we define for each $n\in\mathbb{N}$:
\begin{equation} \label{eq:sn}
S_n := S \cup \{ \mathcal{A} \in \emph{STRUC}[\sigma]: ||\mathcal{A}|| < n \}
\end{equation}
and the family of problems 
\begin{equation}
\mathcal{F}(S) :=  \{ S_n \}_{n \geq 2}.
\end{equation}
\end{definition}

We will prove that $\mathcal{F}(\text{2CC})$ is an universal complete family in $\Sigma_2^p$
we first prove   
that 2CC is $(n_k,k)$-universal for some sequence $\set{n_k}_{k\geq 1}$.


\begin{lemma} \label{le:2cc}
$\cc$ is $(2k+1,k)$-universal for every $k\geq 1$.
\end{lemma}
\begin{proof} Recall $\sigma_g=\langle E^2 \rangle$ is the vocabulary for graphs. Let $k$ be a natural number and let $m\geq 2k+1$. 
We need to verify that for every sequence of $m$-consistent literals over $\sigma_g$, let's say,
\begin{equation}\label{eq:conditions}
L_1(u_1,v_1), L_2(u_2,v_2), \ldots , L_k(u_k,v_k),
\end{equation}
there is $m$-consistency in $\cc$. For every $1\leq i \leq k$, $L_i$ is either $E$ or $\neg E$, and $u_i,v_i \in m$. 
These $k$ conditions are consistent if and only if there are no loops and there's no pair $(u,v)$ and indexes $i\neq j$ such that 
\begin{equation}
L_i(u,v) \equiv E(u,v)  \quad \textrm{and} \quad L_j(u,v)\equiv \neg E(u,v).
\end{equation}

The following analysis is done under the supposition that (\ref{eq:conditions}) is a consistent sequence of conditions. 
If every literal in (\ref{eq:conditions}) is positive, that is,
\begin{equation}\label{eq:poscond}
E(u_1,v_1), E(u_2,v_2), \ldots, E(u_k,v_k)
\end{equation}
then the complete graph on $m$ nodes is a model of (\ref{eq:poscond}), but this is also a positive instance of $\cc$, 
because a complete graph has only one maximal clique (itself), and we can choose a coloration in order to obtain a nonmonochromatic complete graph. 

If there are negative literals in (\ref{eq:conditions}), we can rearrange the sequence so that every negative literal appears at the end:
\begin{equation}\label{eq:negcond}
E(u_1,v_1), \ldots,E(u_j,v_j),\neg E(u_{j+1},v_{j+1}),\ldots , \neg E(u_k,v_k),
\end{equation}
for some $0\leq j \leq k-1$ (if $j=0$, that means every literal in the sequence is negative). 
Let $\mathcal{G} = \langle m,E^{\mathcal{G}} \rangle$ be the \textit{biggest} graph that satisfies (\ref{eq:negcond})
i.e.
\[
 E^{\mathcal{G}}=\setdef{(a,b)\in m\times m}{a\not=b\text{ and } (a,b)\not=(u_i,v_i)\text{ for }j<i\leq k}
\]
The graph $\mathcal{G}$ is a positive instance of $\cc$. The following coloration will certify it. 
Let $R$ be the set of every node that does not appear in a negative condition in sequence (\ref{eq:negcond}), that is,
\begin{equation}
R = \{ x\in m : x\not\in \{ u_i,v_i \} \textrm{ for every } j<i\leq k \}.
\end{equation}
Vertices in $R$ are red and vertices in $R^c$ are blue. 

The set $R$ is nonempty because in the worst case (when $j=0$) there might be at most $2k$ different nodes affected by (\ref{eq:negcond}), 
and this set cannot be $m$ either because we are assuming there are negative conditions. 
Notice that the subgraph induced by $R$ is a clique and every other red clique is completely contained in $R$. 

Now the subgraph induced by $R$ is the maximal red clique, but it is not a maximal clique since there are edges joining every vertex in $R$ with every vertex in $R^c$. This very same argument shows there is no blue maximal clique since every blue clique is contained in $R^c$. Hence, there is no maximal clique with all vertices with the same color, which means that sequence (\ref{eq:negcond}) is consistent in $\cc$.
\qed 
\end{proof}

\begin{corollary}\label{cor:universality}
Given any natural number $n\geq 2$,
$\cc_n$ is $(2k+1,k)$-universal for every $k\geq 1$.
\end{corollary}

\begin{proof}
Direct from Lemmas \ref{lemma:universal} and \ref{le:2cc} and
the fact that $\cc \subseteq \cc_n$ for every $n\in\mathbb N$. \qed
\end{proof}

\begin{lemma}\label{le:belong}
For every natural number $n\geq 2$ the problem 
$\cc_n$ belongs to $\Sigma_2^p$.
\end{lemma}

\begin{proof}
Padding a problem as in equation (\ref{eq:sn}) does not affect its complexity: if $S$ is in $\Sigma_k^p$, then it has a defining sentence $\Phi$ in $\textrm{SO}_k$. For every $\ell$ consider the FO sentence
\[
\varphi_\ell:=\exists x_1\ldots\exists x_\ell\,\forall y\;
\left[\bigwedge_{1\leq i< j\leq \ell} x_i\not= x_j\right]\land
\left[\bigvee_{1\leq i\leq \ell} y=x_i\right]
\] 
It is clear that a structure satisfies $\varphi_\ell$ iff its cardinality is exactly $\ell$.
Thus the sentence 
\[
\Phi\lor\left[\bigvee_{1\leq \ell<n}\varphi_\ell\right]
\]
defines $S_n$ and it is a sentence in $\textrm{SO}_k$ so
$S_n$ is still in $\Sigma_k^p$ for every $n\in\mathbb{N}$. Therefore, $\mathcal{F}(S)$ is a family in $\Sigma_k^p$. \qed
\end{proof}

\begin{lemma}\label{le:2ccnHardness}
For every natural number $n\geq 2$ the problem 
$\cc_n$ is $\Sigma_2^p$-hard.
\end{lemma}

\begin{proof}
For every $n \in \mathbb{N}$ we define a fop $\rho_n$ such that
for every simple graph $\mathcal{G}$
\begin{equation}\label{eq:rhon}
\mathcal{G} \in \cc \quad \iff \quad \rho_n(\mathcal{G}) \in \cc_n.
\end{equation}

Given any graph $\mathcal{G} \in \cc$ we want its image $\rho_n(\mathcal{G})$ to have cardinality at least $n$, since otherwise it belongs to $\cc_n$ by definition. 
The reduction will consist in \textit{padding} $\mathcal{G}$ keeping its basic structure as shown in the image below. 

        Let $k$ be the minimum integer such that $2k>n$. This is enough because every structure has at least two elements. 
        For every simple graph $\mathcal{G}$ its image $\rho_n(\mathcal{G})$ consists of $k$ disconnected copies of $\mathcal{G}$.
Notice that any maximal clique in $\rho_n(\mathcal{G})$ has exactly the same cardinality as a maximal clique in $\mathcal{G}$ and
that any coloring of the vertices in $\rho_n(\mathcal{G})$ is obtained by $k$ independent colorings of the vertices in $\mathcal{G}$, so this map clearly satisfies property (\ref{eq:rhon}). 

\

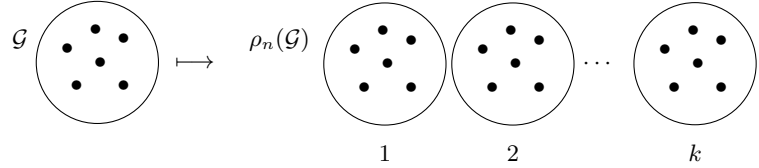
\begin{figure}
\begin{center}
		\begin{tikzpicture}[scale=1.25]
			\node (O) at (0,0){$\bullet$};
			\node (A) at (0.5,0.5){$\bullet$};
			\node (B) at (0.25,0.25){$\bullet$};
			\node (C) at (0.5,0){$\bullet$};
			\node (D) at (0.2,0.6){$\bullet$};
			\node (E) at (-0.1,0.4){$\bullet$};
			\node (G) at (-.6,.5){$\mathcal{G}$};
			\node (Num) at (.22,-.8){$ $};
			\draw (.22,.25) circle (.65cm);
		\end{tikzpicture}
			\begin{tikzpicture}[scale=1.25]					
			\node (Arr) at (-1.8,0.25){$\longmapsto$};
			\node (O) at (0,0){$\bullet$};
			\node (A) at (0.5,0.5){$\bullet$};
			\node (B) at (0.25,0.25){$\bullet$};
			\node (C) at (0.5,0){$\bullet$};
			\node (D) at (0.2,0.6){$\bullet$};
			\node (E) at (-0.1,0.4){$\bullet$};
			\node (G) at (-.9,.5){$\rho_n(\mathcal{G})$};
			\node (Num) at (.22,-.7){$1$};
			\draw (.22,.25) circle (.65cm);
			\end{tikzpicture}\begin{tikzpicture}[scale=1.25]				
			\node (Arr) at (-.4,0.25){$ $};
			\node (O) at (0,0){$\bullet$};
			\node (A) at (0.5,0.5){$\bullet$};
			\node (B) at (0.25,0.25){$\bullet$};
			\node (C) at (0.5,0){$\bullet$};
			\node (D) at (0.2,0.6){$\bullet$};
			\node (E) at (-0.1,0.4){$\bullet$};
			\node (Num) at (.22,-.7){$2$};
			\draw (.22,.25) circle (.65cm);
			\end{tikzpicture}\begin{tikzpicture}[scale=1.25]
			\node (Arr) at (-.8,0.25){$\cdots$};
			\node (O) at (0,0){$\bullet$};
			\node (A) at (0.5,0.5){$\bullet$};
			\node (B) at (0.25,0.25){$\bullet$};
			\node (C) at (0.5,0){$\bullet$};
			\node (D) at (0.2,0.6){$\bullet$};
			\node (E) at (-0.1,0.4){$\bullet$};
			\node (Num) at (.22,-.7){$k$};
			\draw (.22,.25) circle (.65cm);
			\end{tikzpicture}
            \caption{Reduction $\rho_n$, that assigns to every simple graph $\mathcal{G}$ $k$ copies of itself}
\end{center}
		\end{figure}

It is easy to see that $\rho_n$ is a projection. The arity of $\rho_n$ can be settled as $a=\log(k-1)+2$.
The elements of $|\rho_n(\mathcal{G})|$ are $a$-tuples $\vec{u}=(u_0,\ldots, u_{a-1})$
where the first $a-1$ coordinates $u_0,\ldots, u_{a-2}$ code in binary a number between $0$ and $k-1$ 
which identifies a copy of $\mathcal{G}$ and $u_{a-1}$ is any element of $|\mathcal G|$.
We have an edge between $\vec{u}=(u_0,\ldots, u_{a-1})$ and $\vec{v}=(v_0,\ldots, v_{a-1})$ if and only if
$u_0,\ldots, u_{a-2}$ and $v_0,\ldots, v_{a-2}$ represent the same number in binary and $(u_{a-1},v_{a-1})$
is an edge in $\mathcal G$.

Formally, consider $\rho_n$ as a map from $\struc[E]$ to $\struc[Q]$ where $E$ and $Q$ are binary relation symbols.
Consider first formulas
\[
\theta_j(x_0,\ldots, x_{a-1}):=x_0=\ell_0\land\ldots\land x_{a-2}=\ell_{a-2}
\]
where $\ell_i$ is 0 or 1 according to the digit in the corresponding position of the binary representation of $j$
i.e.
\begin{alignat*}{2}
\theta_0(\vec x) & :=x_0=0\land x_1=0 \land\ldots\land x_{a-3}=0\land x_{a-2}=0\\
\theta_1(\vec x) & :=x_0=0\land x_1=0 \land\ldots\land x_{a-3}=0\land x_{a-2}=1\\
\theta_2(\vec x) & :=x_0=0\land x_1=0 \land\ldots\land x_{a-3}=1\land x_{a-2}=0\\
\theta_3(\vec x) & :=x_0=0\land x_1=0 \land\ldots\land x_{a-3}=1\land x_{a-2}=1\\
	& \vdots
\end{alignat*} 
and so on.
Then
\[
\varphi_0(\vec x):= \theta_0(\vec x)\lor\ldots\lor \theta_{k-1}(\vec x)
\]
is a numeric formula and defines the universe of $\rho_n(\mathcal{G})$ and 
\[
\varphi_1(\vec x, \vec y):= (x_0=y_0\land\ldots\land x_{a-2}=y_{a-2})\land E(x_{a-1},y_{a-1})
\]
defines the binary relation $Q^{\rho_n(\mathcal{G})}$.

 \qed
\end{proof}

\begin{theorem}\label{theorem:family} $\mathcal{F}(\cc)$ is a complete and universal family in $\Sigma_2^p$.
\end{theorem}
\begin{proof} 
By Corollary \ref{cor:universality} every problem in $\mathcal{F}(\cc)$ is $(2k+1,k)$-universal for every natural number $k\geq 1$.
By Lemma Lemma \ref{le:belong} every problem in $\mathcal{F}(\cc)$ belongs to $\Sigma_2^p$ and by Lemma \ref{le:2ccnHardness} every problem in $\mathcal{F}(\cc)$ is $\Sigma_2^p$-hard hence every problem in the family is $\Sigma_2^p$-complete.

Therefore $\mathcal{F}(\cc)$ is a complete and universal family in $\Sigma_2^p$.

\qed
\end{proof}

As a consequence of \textbf{Theorems \ref{theorem:borges}} and \textbf{\ref{theorem:family}} we have the following:

\begin{theorem}\label{theorem:superfluous}
$\forall\mathrm{FO}$ is superfluous with respect to $\Sigma_2^p$.
\end{theorem}

We proceed now to show the practicality of the last result. 
Let the vocabulary $\tau = \langle P^2,N^2,V^2,K^1 \rangle$. A $\tau$-structure $\mathcal{A}$ is an instance of $\vcsat$ when a proper interpretation of $\tau$ symbols is given. For every $i,j \in |\mathcal{A}|$,
\begin{itemize}
\item $\mathcal{A} \models P(i,j)$ iff Boolean variable $j$ appears positively in  implicant $i$ of $\mathcal{A}$.
\item $\mathcal{A} \models N(i,j)$ iff Boolean variable $j$ appears negatively in implicant $i$ of $\mathcal{A}$.
\item $\mathcal{A} \models V(i,j)$ iff there is a 1 in bit $j$ in the binary codification of $v_i$.
\item $\mathcal{A} \models K(i)$ iff there is a 1 in bit $i$ in the binary codification of the cost.
\end{itemize}

First two items means that $\mathcal{A}$ is a Boolean formula in DNF. Let $\Psi$ be the $\mathrm{SO}_2[\tau]$ formula that characterizes $\vcsat$, which we need  to prove the hardness of this problem. 

\begin{proposition}
$\vcsat$ is $\Sigma_2^p$-hard. 
\end{proposition}

\begin{proof} Given an instance of $\vcsat$ we want to know if it can be reconsidered as a positive instance of $\efsat$ by defining a set of existential variables under certain cost. 

With the following first order formulas, we make irrelevant the calculation of the total value of possible existential variables.
\begin{alignat*}{2}
\psi_1 & := \forall xy ((V(x,y) \leftrightarrow y=0) \oplus V(x,y)) \\
\psi_2 & := \forall x (K(x) \leftrightarrow x=\max)
\end{alignat*}  

The symbol $\oplus$ refers to the exclusive disjunction. If $\mathcal{A}$ is a $\tau$-structure and $n$ its cardinality, then
\begin{itemize}
\item $\mathcal{A} \models \psi_1$ iff the value of every Boolean variable is either 1 or $2^n-1$.
\item $\mathcal{A} \models \psi_2$ iff the cost is $2^{n-1}$.
\end{itemize}

The calculation of total value is irrelevant because with $n$ variables only those with value 1 can be consider to be existential (the rest have a value that exceed the actual cost by much) and even if the value of all the variables is 1, 
$$n \leq 2^{n-1}, \quad \textrm{for all } n \geq 2.$$

Consider the problem characterized by the conjunction $\Psi \wedge \psi_1 \wedge \psi_2$. An instance of this \textit{new} problem is not only an instance of $\vcsat$, it is one for which is quite clear to determine the set of existential variables. Once this set is defined, the only property that remains uncertain is $\efsat$. If we restrict the instances to those that satisfy $\psi_1 \wedge \psi_2$, problem $\mathrm{MOD}[\Psi \wedge \psi_1 \wedge \psi_2]$ is precisely $\efsat$, which is $\Sigma_2^p$-hard. By \textbf{Theorem \ref{theorem:superfluous}} $\vcsat$ is $\Sigma_2^p$-hard because $\psi_1\wedge\psi_2$ is universal. 

\qed 
\end{proof}


\subsection{Superfluity in $\Pi_2^p$}

Although we already know that superfluity can be used in $\Sigma_2^p$, we can not conclude by duality that the method of superfluity can also be applied in $\Pi_2^p$, for which we need to construct a complete and universal family in this complexity class from scratch. However, the problem $(\cc)^c$ (the complement of $\cc$) is a good candidate to represent an universal problem. We already know that $(\cc)^c$ is $\Pi_2^p$-complete (in this case, a duality argument is valid). We need to ensure that there are sufficient instances in $(\cc)^c$. 

\begin{lemma}
$(\cc)^c$ is $(2k+5,k)$-universal for every $k \geq 1$.
\end{lemma}
\begin{proof}
Before attempting to prove that any sequence of $k$ consistent conditions over vocabulary $\sigma_g$ is consistent in $(\cc)^c$, notice that the \textit{smallest} graph (the graph with minimum nodes and edges) that satisfies $(\cc)^c$ is the five-node cycle. Exhaustively can be checked that any other graph with less nodes or edges is in $\cc$. 

Now, let $L_1(u_1,v_1),\ldots,L_k(u_k,v_k)$ be $m$-consistent literals (like in the proof of \textbf{Theorem \ref{le:2cc}}) and consider the smallest graph $\mathcal{G}$ on $m$ nodes that satisfies all $k$ conditions, that is, $\mathcal{G} = \langle m,E^{\mathcal{G}} \rangle$ where
\begin{equation*}
E^{\mathcal{G}} = \big\{ (u,v)\in m\times m: \{u,v\}=\{u_i,v_i\} \textrm{ for some }i=1,\ldots,k \textrm{ and } L_i=E \big\}.
\end{equation*}

$\mathcal{G}$ might not be an instance of $(2\textrm{CC})^c$ because with the $k$ conditions a positive instance of 2CC can be constructed, but notice that as $m \geq 2k+5$ there are always at least five free nodes. With this five nodes a five-node cycle $\mathcal{C}$ can be constructed and attached to $\mathcal{G}$ to create a new graph called $\mathcal{G}'$. The precise structure of this new graph is $\langle m,E^{\mathcal{G}'} \rangle$ where
\begin{equation*}
E^{\mathcal{G}'} = E^{\mathcal{G}} \cup E^{\mathcal{C}}.
\end{equation*}

As $\mathcal{C}$ is a maximal connected subgraph of $\mathcal{G}'$, there is no way to color the nodes of $\mathcal{G}'$ such that every maximal clique could be nonmonochromatic, because at least one of the edges of $\mathcal{C}$ (which turn out to be maximal cliques) is monochromatic for every coloration. \qed
\end{proof}

Using a similar argument as in the proof of \textbf{Theorem \ref{theorem:family}} we can conclude the following:

\begin{theorem}
$\mathcal{F}((\cc)^c)$ is a complete and universal family in $\Pi_2^p$.
\end{theorem}

\begin{theorem}
$\forall\mathrm{FO}$ is superfluous with respect to $\Pi_2^p$.
\end{theorem}

\section{Conclusions}\label{sec:Conclusions}

\textbf{Theorem \ref{theorem:borges}} is used in \cite{borges3} to prove that $\forall\textrm{FO}$ is superfluous with respect to the complexity classes {\bf NL}, {\bf P}, {\bf NP} and {\bf coNP}. We have enlarged that list proving that the superfluity method is also applicable in the complexity classes corresponding to the Second Level of the Polynomial-Time Hierarchy, $\Sigma_2^p$ and $\Pi_2^p$. 

As \textbf{Definition \ref{def:family}} is strongly semantical, there is still no generic proof of superfluity in every level of {\bf PH}, but we believe it is the case. One way to tackle this problem might be considering a sequence $\{A_k\}_k$ of $\Sigma_k^p$-complete problems, with an intrinsic relation on the vocabularies involved. $\cc$ is an appropriate problem to study the universal property because only one relation symbol is required to express it through second-order logic. It would be ideal that every problem in the sequence $\{A_k\}_k$ can be represented as easily as $\cc$. If we manage to generalize $\cc$ to a $\Sigma_k^p$-complete version that preserves universality for every $k$, we can consider superfluity in every level of {\bf PH} solved. 

Immerman-Medina conjecture is still a source for future investigation, since no other fragments of FO has been proven superfluous. A superfluity version of $\exists\textrm{FO}$ with respect to {\bf NP} might lead to an inductive proof of superfluity for FO. 

%
%

\appendix

\section{Problems in $\Sp$ referred to in this paper}\label{app:problems}
\begin{enumerate}
	\item{\sc $2-$Quantified Satisfiability} ($\efsat$)
	\begin{description}
	\item[Instance:]  A DNF Boolean formula $\phi(\mathbf{x},\mathbf{y})$, 
		where $\mathbf{x}$ and $\mathbf{y}$ are tuples of Boolean variables of length $n$ and 
		$m$ respectively.
	\item[Question:]Is there a vector $\mathbf{x} \in \{0,1\}^n$ such that for every vector 
		$\mathbf{y} \in \{0,1\}^m$, $\phi(\mathbf{x},\mathbf{y})$ is true?
	\item[Notes:] The vocabulary used to codify Boolean formulas (with existential variables) as finite structures is
			$\sigfnd=\tup{E^1,Q^2,M^2}$.	
	\end{description}
    
     \
    
	\item{\sc $2-$Quantified Unsatisfiability} ($\eunsat$)
	\begin{description}
	\item[Instance:]  A CNF Boolean formula $\phi(\mathbf{x},\mathbf{y})$, 
		where $\mathbf{x}$ and $\mathbf{y}$ are tuples of Boolean variables of length $n$ and 
		$m$ respectively.
	\item[Question:]Is there a vector $\mathbf{x} \in \{0,1\}^n$ such that for every vector 
		$\mathbf{y} \in \{0,1\}^m$, $\phi(\mathbf{x},\mathbf{y})$ is false?
	\item[Notes:]  The vocabulary used to codify Boolean formulas (with existential variables) as finite structures is
			$\sigfnc=\tup{E^1,P^2,N^2}$.	
	\end{description}
    
     \
        
	\item{\sc Unique Extension Satisfiability} ($\eunique$)
	\begin{description}
	\item[Instance:]  A CNF Boolean formula $\phi(\mathbf{x},\mathbf{y})$, 
		where $\mathbf{x}$ and $\mathbf{y}$ are tuples of Boolean variables of length $n$ and 
		$m$ respectively.
	\item[Question:]Is there a vector $\mathbf{x} \in \{0,1\}^n$ such that for an unique vector 
		$\mathbf{y} \in \{0,1\}^m$, $\phi(\mathbf{x},\mathbf{y})$ is true?
	\end{description}
    
     \
	
	\item {\sc $2-$Clique Coloring} ($2\textrm{CC}$)
	\begin{description}
	\item[Instance:] A simple graph $\mathcal{G} = \langle V,E \rangle$.
	\item[Question:] Is there a $2$-clique-coloring of $\mathcal{G}$, that is,
			a $2$-coloring of $V$ such that every maximal clique (complete subgraph) of $\mathcal{G}$ 
			is nonmonochromatic?
	\item[Notes:]The vocabulary used to codify graphs is naturally the one that consists of one 
		binary relation:
		$\sigma_g = \langle E^2 \rangle$. In some cases another binary symbol is used to avoid misunderstanding.
	\end{description}
    
     \
	
	\item {\sc Value-Cost Satisfiability} ($\vcsat$)
	\begin{description}
	\item[Instance:] A DNF Boolean formula $\phi$ on Boolean variables $x_1,\ldots,x_n$, 
			an integer $K$ and an integer value $v_i$ for each variable $x_i$. 
	\item[Question:] Is there a choice of Boolean variables with total value below $K$ and 
			such that $\phi$ along with that choice is a $\efsat$ instance?
	\end{description}
\end{enumerate}
\end{document}